\documentclass[letterpaper,twocolumn,10pt]{article}

\usepackage{mathptmx}
\usepackage[T1]{fontenc}
\usepackage[utf8]{inputenc}
\usepackage{pslatex}
\usepackage[kerning,spacing]{microtype}

\usepackage[dvipsnames]{xcolor}
\usepackage{bm}
\usepackage{algorithm}
\usepackage[noEnd=True]{algpseudocodex}
\usepackage{tikz}
\usepackage{pgfplots}
\usetikzlibrary{
    pgfplots.dateplot,
}
\usepackage{booktabs}
\usepackage{mathtools}
\usepackage{annotate-equations}
\usepackage{amsthm}
\newtheorem{theorem}{Theorem}
\usepackage{svg}
\usepackage{amssymb,amstext} 
\usepackage{array}
\usepackage{xurl}
\usepackage{breakurl}
\usepackage[]{hyperref}
\hypersetup{
  colorlinks,
  linkcolor={green!80!black},
  citecolor={red!70!black},
  urlcolor={blue!70!black}
}
\usepackage{framed}
\colorlet{shadecolor}{red!10}

\begin{document}

\date{}

\title{\Large \bf Over-Threshold Multiparty Private Set Intersection for Collaborative Network Intrusion Detection}

\author{
{\rm Onur Eren Arpaci}\\
University of Waterloo
\and
{\rm Raouf Boutaba}\\
University of Waterloo
\and
{\rm Florian Kerschbaum}\\
University of Waterloo
}

\maketitle

\begin{abstract}
An important function of collaborative network intrusion detection is to analyze the network logs of the collaborators for joint IP addresses. However, sharing IP addresses in plain is sensitive and may be even subject to privacy legislation as it is personally identifiable information. In this paper, we present the privacy-preserving collection of IP addresses. We propose a single collector, over-threshold private set intersection protocol. In this protocol $N$ participants identify the IP addresses that appear in at least $t$ participant's sets without revealing any information about other IP addresses. Using a novel hashing scheme, we reduce the computational complexity of the previous state-of-the-art solution from $\mathcal{O}(M(N \log{M}/t)^{2t})$ to $\mathcal{O}(t^2M\binom{N}{t})$, where $M$ denotes the dataset size. This reduction makes it practically feasible to apply our protocol to real network logs. We test our protocol using joint networks logs of multiple institutions. Additionally, we present two deployment options: a collusion-safe deployment, which provides stronger security guarantees at the cost of increased communication overhead, and a non-interactive deployment, which assumes a non-colluding collector but offers significantly lower communication costs and applicable to many use cases of collaborative network intrusion detection similar to ours.
\end{abstract}

\section{Introduction}
Many cyberattacks are coordinated and target multiple victims. An analysis from the security company Risk Analytics shows that 75\% of attacks on institutions spread to a second institution within one day, and over 40\% spread within one hour \cite{jasper2017us}. It is difficult for standalone intrusion detection systems to identify large scale threats due to lack of contextual information \cite{li2022surveying}. This underscores the need for collaborative network security, where institutions share and compare security data (logs, alerts, indicators) to quickly identify common threats.

A considerable body of work exists on collaborative threat detection~\cite{zhou2007evaluation, zhou2010survey, li2022surveying} as well as privacy preserving collaborative threat detection \cite{lincoln2004privacy, locasto2005towards, davy2023privacy}. However, until now the research focused on sharing alerts, anomalies, and indicators with other institutions. This approach assumes that the individual intrusion detection systems locally filter the raw network data such as IP logs, connections etc., and only report the "out-of-ordinary" events to other institutions. Zabarah et al.~\cite{zabarah2023approach} demonstrated that analyzing raw network data across multiple institutions can reveal attacks that would go undetected by individual institutions alone. However, the biggest challenge for this approach is privacy. The problem is twofold. First, the raw data is significantly more sensitive than security alerts because it contains personally identifiable information. Second, the raw data is orders of magnitude larger than security alerts, making any existing solution computationally unfeasible. We propose a protocol that addresses both of these problems.

Zabarah et al.~\cite{zabarah2023approach} found that a large portion of cyberattacks are typically initiated from a small number of IP addresses that are external to the institutions. Moreover, an external IP address initiating a connection to the institution's internal nodes is an atypical behavior if this connection is not for a publicly advertised service such as a website. Following these characteristics, Zabarah et al. propose a simple criterion with 95\% recall rate for detecting multi-institution attacks. If an external IP connects to at least $t$ institutions within a specific time window, it is classified as malicious, where $t$ is an empirically decided threshold. This solution is simple and effective. However, it requires institutions to share their network logs, which creates the aforementioned privacy concern.

We address the privacy concern with an efficient Over-Threshold Multiparty Private Set Intersection (OT-MP-PSI) protocol. OT-MP-PSI is a cryptographic problem where $N$ participants each hold a set that contains at most $M$ elements. These participants want to know which elements appear in at least threshold $t$ number of sets without revealing any information about elements that do not meet this threshold. Figure~\ref{fig:otmpsivisualization} shows a visualization of the problem. This problem is a direct abstraction of the Zabarah et al.'s~\cite{zabarah2023approach} approach. The institutions do not want to share possibly benign IP addresses, but they are willing to share the IP addresses appearing in multiple institutions, and are likely malicious.

While OT-MP-PSI has other applications such as heavy hitter identification within the network~\cite{kamiyama2007simple, liu2016identifying, cao2009identifying}, private file deduplication on cloud~\cite{jan14dedup} or identifying high-risk individuals in the spread of disease~\cite{bay2022practical}, this paper uses the problem of multi-institution attacks as the main focus and subject of evaluation.

\begin{figure}[h]
    \centering
    \includegraphics[width=0.475\textwidth]{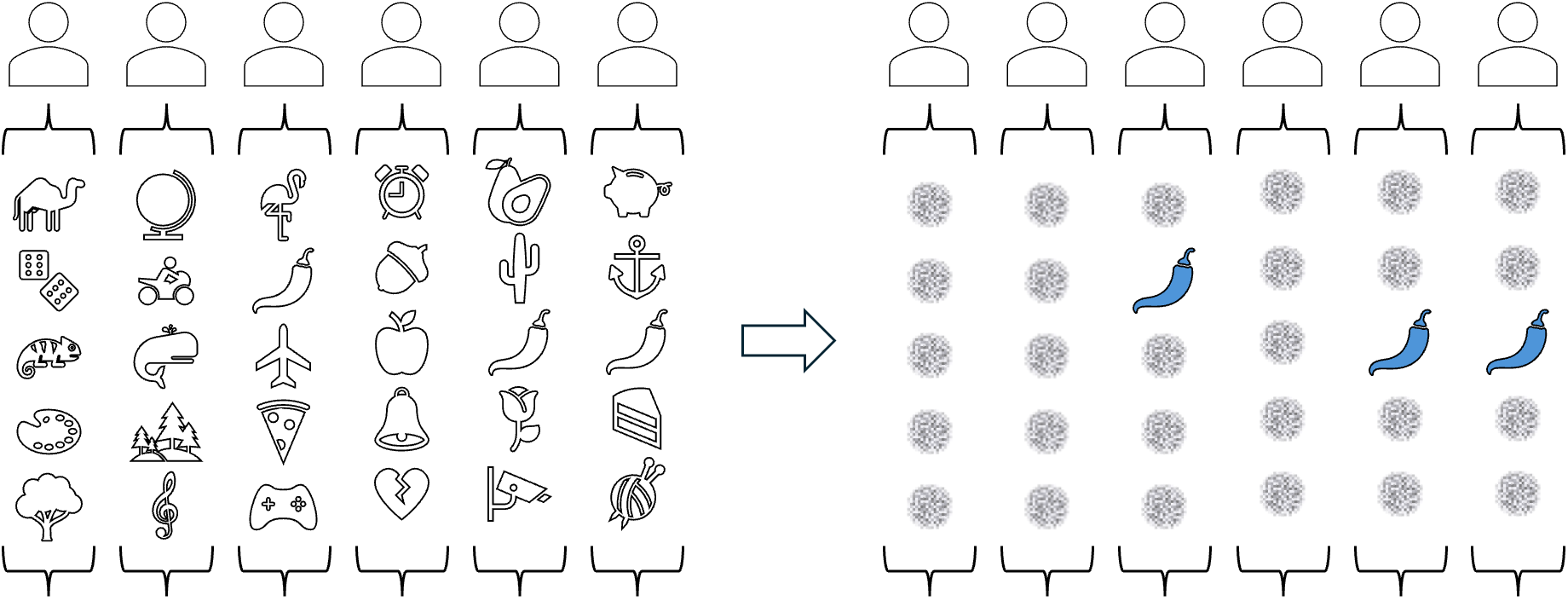}
    \caption{The OT-MP-PSI protocol visualization for parameters ${N=6}, {M=5}, {t=3}$}
    \label{fig:otmpsivisualization}
\end{figure}

Kissner and Song proposed the first solution to the OT-MP-PSI problem~\cite{kissner2004private}. However, their solution requires a high number of communication rounds and it is computationally expensive. Specifically, their protocol requires $\mathcal{O}(N)$ communication rounds, has $\mathcal{O}(N^3M)$ communication complexity, and $\mathcal{O}(N^3M^3)$ computation complexity. This makes it impractical for real-life applications with large data sets because it requires each participant to have significant computational resources and to be online for the duration of the protocol. Additionally, the sequential nature of the protocol hinders the scalability as adding more compute resources does not directly result in faster execution times.

Mahdavi et al.~\cite{mahdavi2020practical} introduce a solution with a constant number of rounds and better communication complexity. However, their scheme still imposes impractically high computational costs for our use case.

In this work, we propose an OT-MP-PSI protocol that outperforms the previous state-of-the-art solution~\cite{mahdavi2020practical} by factors ranging from $33\times$ to $23,066\times$ in terms of computational overhead. This efficiency improvement enables the practical usage of the protocol in the collaborative network intrusion detection problem because the nature of the problem requires the processing of large IP address sets as each institution receives connections from hundreds of thousands of IP addresses every hour. In a setting with 33 participating institutions and at most 144,045 IPs connecting to any institution, the system can detect malicious IPs in 170 seconds, as we show in our experiments. The previous state-of-the-art solution~\cite{mahdavi2020practical} requires multiple days to complete the same computation.

We use Shamir's secret sharing~\cite{shamir} to reveal identical items that are over the threshold. Shamir's secret sharing allows a secret data item to be split into many shares such that any size-$t$ subset of the shares can be used to reconstruct the original secret. However, subsets smaller than $t$ do not reveal anything about the original secret. We make every participant create a single secret share for each element they own, and then combine these secret shares among participants to find the elements that exist in $t$ or more sets. The main challenge with this approach is that all $t$ combinations of secret shares must go through the reconstruction process to find the shares corresponding to the same set element because secret shares do not reveal anything about the underlying set element. A naive approach requires $\binom{N}{t} M^t $ different combinations. Mahdavi et al.~\cite{mahdavi2020practical} uses a binning technique that reduces the necessary combinations to $\mathcal{O}(M(N \log{M}/t)^{2t})$. As our main contribution, we introduce a new hashing scheme that reduces this number to $\mathcal{O}(tM\binom{N}{t})$ linear in $M$ and hence scalable to real problem sizes.

Moreover, we show that our protocol efficiently handles the case where the threshold is equal to the number of participants $(t=N)$ with a computational complexity of $O(N^{2}M)$. This problem is of independent interest \cite{bay2022practical, morales2023private}.

We prove the protocol's security under the semi-honest multiparty computation model. The semi-honest model prevents many practical attacks, such as eavesdropping by participants, but can be extended to the malicious model to handle active adversaries. However, even under the malicious model, private set intersection protocols remain vulnerable to input substitution attacks, which can reveal targeted information.

For network efficiency and ease of deployment, the protocol uses a dedicated Aggregator that combines the secret shares with some negligible leakage about the contents of the sets. We will explain precisely what this leakage is in Section \ref{sec:usecase}.

We offer two separate deployment options for our protocol: collusion-safe option and non-interactive option. Collusion safe option requires constant rounds of interactions and has more communication overhead but it resists collusions. The non-interactive option has a much simpler structure and less communication overhead but it assumes the aggregator is non-colluding.

We implement our protocol and evaluate its performance on real-world network connection logs from CANARIE IDS program~\cite{canarie} to show its practicality.

The remainder of this paper is structured as follows. We define the security model and describe some preliminary cryptographic concepts in Section~\ref{sec:preliminaries}. We explain the use case of the protocol in detail and provide the formal functionality of the protocol in Section~\ref{sec:usecase}. We present the main challenges, our approach, and the formal description of our protocol in Section~\ref{sec:protocol}. We give a formal analysis of the failure probability of our new hashing scheme in Section~\ref{sec:hashing}. We prove our protocol's security and show its efficiency with both theoretical analysis and experimental results in Section~\ref{sec:eval}. We review related work in Section~\ref{sec:related} before we conclude our findings and discuss future work in Section~\ref{sec:conclusion}.

\section{Preliminaries}
\label{sec:preliminaries}
\subsection{Security Model}
\label{sec:secmodel}
We follow Goldreich's definition where a semi-honest party adheres to the established protocol but records all its intermediate calculations and any messages it receives from other parties \cite{goldreich}. We say a protocol is secure under the semi-honest model if each party cannot gain any other information except the intended output of the protocol. More formally, let $f: (\{0,1\}^*)^n \longmapsto (\{0,1\}^*)^n$ be an n-ary functionality, where $f_i(x_1, ..., x_n)$, denotes the $i^{th}$ element of $f(x_1, ..., x_n)$. Let $\Pi$ be an n-party protocol for computing $f$. The view of the $i^{th}$ party during an execution of $\Pi$ on $\overline{x} = (x_1, ..., x_n)$ is defined as $VIEW_i^{\Pi} = (x_i, m_1, ..., m_t)$ where $m_i$ represents the $i^{th}$ message it has received. We say that $\Pi$ privately computes $f$ with respect to the semi-honest model if there exists a polynomial-time algorithm, denoted $SIM$, such that for every $i \in [n]$
\begin{equation}
\{SIM(x_i, f_i(\overline{x})\}_{\overline{x} \in\left(\{0,1\}^*\right)^n} \stackrel{\mathrm{c}}{\equiv} \{VIEW_i^{\Pi}(\overline{x})\}_{\overline{x} \in\left(\{0,1\}^*\right)^n}
\end{equation}
where $\stackrel{\mathrm{c}}{\equiv}$ denotes computationally indistinguishable.

\subsection{Shamir's Secret Sharing}
The objective of secret sharing is to divide a secret value $V$ into $n$ shares so that it can be reconstructed when a specific subset of these shares is combined. In a $(t,n)$-threshold scheme, any group of $t$ parties can collaborate to retrieve the secret value $V$ using their shares, while any group with fewer than $t$ members cannot gather any information about the secret, regardless of their efforts to collude.

In Shamir’s secret sharing \cite{shamir}, the distributing party generates $t - 1$ values $\{c_i\}_{i \in [t-1]}$ chosen at random from some finite field $F_q$ of prime order $q$ and forms the polynomial
\begin{equation}
\mathbb{P}(x) = c_{t-1}x^{t-1} + c_{t-2}x^{t-2} + ... + c_1x + V
\end{equation}

The distributing party generates $n$ secret shares by evaluating $\mathbb{P}$ at $n$ publicly-known distinct values. For instance, the secret share for party $P_i$ (with $i \in F_q$) is $V_i = (i, f (i))$. Since $t$ points uniquely determine a polynomial of degree $t - 1$, anyone possessing $t$ shares $V_i$ can recover the secret $V$ using Lagrange interpolation:
\begin{equation}
    \label{eq:lagrange}
    V = \sum_{i=1}^{t} \left( V_i \cdot \prod_{j=1, j \neq i}^{t} \frac{-j}{i - j} \right)
\end{equation}

In our protocol, we leverage Shamir's secret sharing to indicate the existence of a set element only if at least $t$ sets include that element.

\subsection{Oblivious Pseudo Random Function (OPRF)}
\label{sec:oprf}

An Oblivious Pseudo-Random Function (OPRF) is a protocol between a key holder that holds a secret key $K$ and a participant that holds an input $x$, where the participant learns the result of the Pseudo-Random Function $\mathbb{H}_K(x)$ without learning anything about the secret $K$, and the key holder does not learn anything about the input $x$ or the value of $\mathbb{H}_K(x)$. In our protocol, we use the following 2HashDH OPRF protocol introduced by Jarecki et al. \cite{jarecki2016}.
Let $H$ and $H'$ be two cryptographic hash functions.
\[
\begin{array}{c c c}
\text{\textsf{\underline{\textbf{Participant}}}} & & \text{\textsf{\underline{\textbf{Key Holder}}}} \\
r \leftarrow_{\text{\tiny{R}}} \mathbb{Z}_q & & \\
a \leftarrow H(x)^r & \xrightarrow{\hspace{2.7em}a\hspace{2.7em}} & b \leftarrow a^K \\
\text{Output: } H'(x,b^{1/r})& \xleftarrow{\hspace{2.7em}b\hspace{2.7em}} & \\
\end{array}
\]
This protocol obliviously evaluates $H'(x,H(x)^K)$. We can extend this protocol to multiple key holders with the following: the participant runs the protocol with the same input for all $k$ key holders and combines the output by multiplying them.
\begin{align}
    H'\left(x,H(x)^{K_1} \cdot H(x)^{K_2} \cdots H(x)^{K_k}\right) \nonumber \\ 
    = H'\left(x,H(x)^{K_1 + K_2 + \dots +K_k}\right) = \mathbb{H}_{K_1\dots K_k}(x) \nonumber
\end{align}

\subsection{Oblivious Pseudo Random Secret Sharing (OPR-SS)}
\label{sec:oprss}
Oblivious Pseudo-Random Secret Sharing (OPR-SS) \cite{mahdavi2020practical} is a protocol that combines share generation and reconstruction properties of secret sharing and the security properties of the OPRFs. It allows participants to get a secret share from key holders, unique for their input and their identity, without key holders learning anything about the input or the secret share, and without participants learning anything about the secret keys of key holders. Later, any $t$ secret shares from $t$ different participants with the same input can be used to reconstruct the original value. The functionality of the OPR-SS is given in Figure \ref{func:opr-ss}.

\begin{figure}[h]
\centering
\fbox{
    \parbox{0.95\linewidth}{
        The functionality is parameterized by a threshold $t$, the number of key holders $k$ and a value $V$.

        \textbf{Input of $\bm{P_i}$:} The participant $P_i$ provides the input $s$

        \textbf{Input of each $\bm{KH_j}$:} Each key holder $KH_j$ provides $t$ secrets $\{K_{j,1},\dots, K_{j,t}\}$. 

        \textbf{Functionality:} 
        The functionality computes the polynomial
        \vspace{-0.1cm}
        \[\mathbb{P}_{s}^{K_1, \dots K_k}(i) = V+\sum_{m=1}^{t-1} i^{m} \cdot H(s)^{K_{1,m} + K_{2,m}+ \dots +K_{k,m}}\]

        \textbf{Output to $\bm{P_i}$:} The secret share $\mathbb{P}_{s}^{K_1, \dots K_k}(i)$. 

        \textbf{Output to each $\bm{KH_j}$:} Nothing.
    }
}
\caption{ Oblivious Pseudo-Randon Secret Sharing (OPR-SS) Functionality $\mathcal{F}_{OPR\!-\!SS}$}
\label{func:opr-ss}
\end{figure}

In our use case $V$ is a public value that is set to $0$. Reconstructing $0$ from $t$ different secret shares will indicate that they correspond to the same input.

\section{Use Case}
\label{sec:usecase}

Detection of multi-institution cyberattacks is a critical issue faced by organizations today. Institutions that share a common context, such as universities within the same country or hospitals in the same city, often become simultaneous targets of malicious actors. Attackers typically leverage vulnerabilities across multiple organizations simultaneously, aiming to maximize their gains before these vulnerabilities are patched. Leveraging this intuition, Zabarah et al.~\cite{zabarah2023approach} introduced an effective indicator with a 95\% recall rate, which counts the number of distinct institutions contacted by an external IP address within a predefined time window. If this count exceeds an empirically determined threshold $t$, the IP address is classified as malicious. This approach is used in practice, due to its simplicity and general applicability, as it requires only monitoring external IP addresses without relying on specific attack signatures.

CANARIE \cite{canarie} is an institutional internet provider that powers the network of 700+ Canadian research and educational institutions. The CANARIE IDS Program~\cite{canarie} is a collaborative initiative between CANARIE, and 106 Canadian research and education institutions. The program collects and analyzes participating institutions’ network logs for intrusion detection purposes. It has deployed the indicator proposed by Zabarah et al.~\cite{zabarah2023approach} and uses it to identify potential threats. However, participating institutions currently send their logs to CANARIE in plaintext so that IP addresses connecting to multiple institutions can be identified. This centralized, plaintext log-sharing model is also common in commercial outsourced Security Operations Centers, so the situation is not unique to our setting. Although straightforward, this method poses significant privacy risks. Institutions are hesitant to share complete network logs because these logs contain personally identifiable information and sensitive operational data unrelated to security threats. The aggregator, learning the entire dataset, obtains far more information than necessary. This also creates a compliance problem with privacy laws around the world. Many jurisdictions, including the European Union’s GDPR~\cite{european_commission_regulation_2016}, Canada’s PIPEDA~\cite{act2000personal}, and California’s CCPA~\cite{ccpa2018}, classify IP addresses as personally identifiable information, limiting their sharing and processing without explicit consent. While these regulations allow exceptions for cybersecurity purposes, such exceptions are permitted only when strictly necessary.

Our protocol addresses these privacy and compliance challenges by enabling institutions to collaboratively identify IP addresses appearing in at least $t$ institutional datasets, indicating potential malicious activity, without disclosing any IP addresses below this threshold. Consequently, institutions can privately and securely detect multi-institution attacks without entrusting a single aggregator with unrestricted access to sensitive data.

We assume that all parties are semi-honest, which means participants adhere to the protocol but may try to learn additional information beyond their intended output. Our protocol offers two deployment options: collusion-safe and non-interactive.

\textbf{Collusion-safe deployment} includes participants, key holders, and the aggregator. Some participants may also serve as key holders. In this topology, participants connect directly to the aggregator in a star configuration, key holders connect to all participants, and at least one key holder connects to all other key holders. This deployment assumes at least one key holder does not collude with the aggregator and requires a constant number of interactions.
In practice, this deployment option is most useful when no neutral third party exists to serve as the Aggregator because it allows one of the participants to act as the Aggregator.

\textbf{Non-interactive deployment} eliminates the key holders and assumes a non-colluding aggregator. Participants may still collude among themselves. This deployment option is useful when there is a neutral, semi-trusted party that can act as the Aggregator, such as our use case with the CANARIE IDS Program~\cite{canarie}.

\begin{figure}[h]
\centering
\fbox{
    \parbox{0.95\linewidth}{

        The functionality is parameterized by a threshold $t$, the number of participants $N$, and number of key holders $k$. Let $P = \{P_1, \dots, P_N\}$ be the set of participants, and $A$ be the aggregator. If using the collusion-safe deployment, let $KH = \{KH_1, \dots, KH_k\}$ be the set of key holders.

        \textbf{Input of each $\bm{P_i}$:} Each participant $P_i$ provides a set $S_i \subseteq S$, where $S$ is the universe of possible elements (e.g., IPv4/IPv6 addresses).

        \textbf{Input of $\bm{A}$:} The aggregator $A$ provides no private input.

        \textbf{Input of each $\bm{KH_j}$:} The key holders provide no private input.

        \textbf{Functionality:} 
        The functionality computes the sets
        \vspace{-0.2cm}
        \[
        I = \{ s \in S \mid s \text{ appears in at least } t \text{ of the sets } \{S_1, \dots, S_N\} \}.
        \]
        \[
        B = \{(b_1, \dots, b_N) \mid 
            \left.\begin{cases} 
            b_i = 1 \text{ if } s \in S_i  \\
            b_i = 0\ \text{ otherwise}
            \end{cases}\right\} \forall{s \in I} \}
        \]

        \textbf{Output to each $\bm{P_i}$:} Each participant $P_i$ receives the set $I \cap S_i $. 

        \textbf{Output to $\bm{A}$:} The aggregator $A$ receives the set $B$.

        \textbf{Output to each $\bm{KH_j}$:} Nothing.
    }
}
\caption{OT-MP-PSI Functionality $\mathcal{F}_{OT\!-\!MP\!-\!PSI}$}
\label{func:otmpsi}
\end{figure}

As shown in Figure \ref{func:otmpsi}, only the over-threshold elements are revealed to the participants and the aggregator only learns the existence of the over-threshold elements among participants. With this functionality, institutions can safely collaborate to detect multi-institution attacks without compromising the privacy of their complete network logs. In our use case of collaborative intrusion detection, the participants identified to be involved in an attack would share the identified potentially malicious IP addresses with other participants and the aggregator through a threat sharing platform such as MISP\footnote{\url{https://github.com/MISP/MISP}}, identify the significant threats with severity estimation and take precautions using next-threat prediction as sugested by Zabarah et. al.~\cite{zabarah2023approach}.

\section{Protocol}
\label{sec:protocol}

\begin{table}[t]\caption{Notations}
    \centering
    \begin{tabular}{r c p{6cm} }
    \toprule
    $P_i$ & $\coloneqq$ & $i^{th}$ Participant\\
    $S$ & $\coloneqq$ & Domain of elements\\
    $S_i$ & $\coloneqq$ & Set of elements held by the participant $P_i$\\
    $s_{i,j}$ & $\coloneqq$ & $j^{th}$ element in the set $S_i$\\
    $A$ & $\coloneqq$ & Aggregator\\
    $H(\cdot)$ & $\coloneqq$ & Cryptographic hash function\\
    $H_K(\cdot)$ & $\coloneqq$ & Hash-based message authentication code (HMAC) with key $K$\\
    $h_K()$ &$\coloneqq$ & The keyed hash function used to map elements into bins \\
    $\mathbb{P}_{s}(x)$ & $\coloneqq$ & Polynomial used for secret share creation\\
    $K$ & $\coloneqq$ & Symmetric secret key used by participants\\
    \multicolumn{3}{c}{}\\
    \multicolumn{3}{c}{\underline{Protocol Parameters}}\\
    \multicolumn{3}{c}{}\\
    $N$ & $\coloneqq$ & Number of participants\\
    $t$ & $\coloneqq$ & Threshold\\
    $M$ & $\coloneqq$ & Maximum number of elements in each set\\
    $k$ & $\coloneqq$ & Number of key holders\\
    \bottomrule
    \end{tabular}
    \label{tab:notations}
\end{table}

We first present a high-level overview of our OT-MP-PSI protocol. The principal idea is that all participants create a secret share for each element in their set and send these shares to the Aggregator. The Aggregator can only reconstruct secrets that have $t$ or more shares. If a successful reconstruction happens, the Aggregator then informs the participants who sent the secret shares about the reconstruction to indicate that the underlying element is in the over-threshold intersection. This approach has a couple of challenges.

\subsection{Using the same polynomial}

We need to coordinate the use of Shamir's secret sharing, but the security of our protocol would be compromised if there was a central party that straightforwardly creates the shares, because the participants would have to disclose their set elements to this party to get a secret share. The OPR-SS protocol described in Section \ref{sec:oprss} is a good solution for this problem. This approach provides resistance to collusion but it is interactive. We also use a simpler approach that is non-interactive but less collusion resistant. In the non-interactive deployment, participants communicate a secret key $K$ among themselves, which is not disclosed to the Aggregator, and they produce the coefficients of the polynomial themselves using a keyed hash function, as described in Equation \ref{eq:poly2}. Here we secret share the value $0$, similar to the OPR-SS protocol. This allows the aggregator to identify the successful reconstructions. Because if the shares correspond to the same element, they reconstruct the number $0$, and if not, they reconstruct a random element of $F_q$.
\begin{equation}
  \label{eq:poly2}
  \mathbb{P}_{s}^{K}(i) = 0 + \sum_{j=1}^{t-1} H_{K}^{j}(s)i^{j}
\end{equation}
Where $s$ is the set element, $i$ is the identifier of the participant $P_i$, $H$ is the hash-based message authentication code (HMAC) function, $K$ is the secret key, and the superscripts of $H_K$ indicate iteration, i.e., ${H^{i}_K(s) = H_K(H^{i-1}_{K}(s))}$. In our use case, the protocol uses IPv4/IPv6 addresses directly as the domain of elements without any preprocessing or mapping. With this approach, the Aggregator does not learn more than the intended output of the protocol as long as it is non-colluding.

\subsection{Reconstruction Complexity}
If we naively send all secret shares to the Aggregator, the Aggregator would need to try $\binom{N}{t} M^t$, an exponential number in $t$, different combinations to find all matching shares in all cases. We cannot use polynomial-time decoding algorithms for error-correcting codes, such as the list decoding~\cite{reed1960polynomial} or Berlekamp-Welch' algorithm~\cite{welch1986error}, because they may fail, if the number of elements exceeds the threshold $t$, but does not reach the decoding threshold for the error-correcting code. Comparing exponentially many combinations is infeasible if the set sizes $M$ are large, even if the exponent $t$ is rather low. We need to give the Aggregator some hint on how to combine the shares, but if we just put an identifier (e.g. keyed hash of the element) alongside the secret share, then the protocol would leak the similarity distribution of the sets and the Aggregator could see which sets have how many overlapping elements, even when those elements are under the threshold. This is more information than the intended output of our protocol.

A common approach to this problem is to put the shares into bins using a hash table. In this way, the Aggregator has to only try $t$-sized combinations of shares that are from the same bins. However, in order to not leak any distribution information, participants have to pad all of the bins with dummy shares so that all of the bins are equal in size. This approach reduces the complexity to $O(M(N \log{M}/t)^{2t})$, as shown by Mahdavi et al.~\cite{mahdavi2020practical}. Empirically, the maximum bin size, corresponding to the $\log{M}$ term in the complexity, carries large constants. This prevents the system from scaling to our problem size.

As the main contribution of this paper, we introduce a new hashing algorithm. We use hash tables (hereafter, simply tables) with bins of size 1 to map the secret shares. Instead of accommodating hash collisions with bigger bin sizes, we select one of the secret shares that map into the same bin using a pseudo-random ordering and only put that secret share into the table. Consequently, a single table will be missing some of the secret shares. Each participant creates multiple such tables. The motivation is that if participants create enough tables, missing values in one table will appear in others, and the probability of failure can be bounded to an acceptably low value with minimal cost. The main advantage of this approach is that the Aggregator does not need to try combinations of shares. It only needs to try combinations of participants. After selecting $t$ participants, the Aggregator applies Lagrange interpolation to the shares corresponding to the same table and bins (e.g., the first participant's first table's first bin and the second participant's first table's first bin). An illustration of the new hashing scheme is shown in Figure~\ref{fig:hashingvisualization}. We show that if the size of each table is set to $t \times M$ and each participant produces 20 tables (with some additional optimizations), the probability of getting an incorrect result is smaller than $2^{-40}$. We explain this hashing scheme in detail and show the failure probability in Section~\ref{sec:hashing}.
\vspace{-0.2cm}
\begin{figure}[h]
  \centering
  \includegraphics[width=0.475\textwidth]{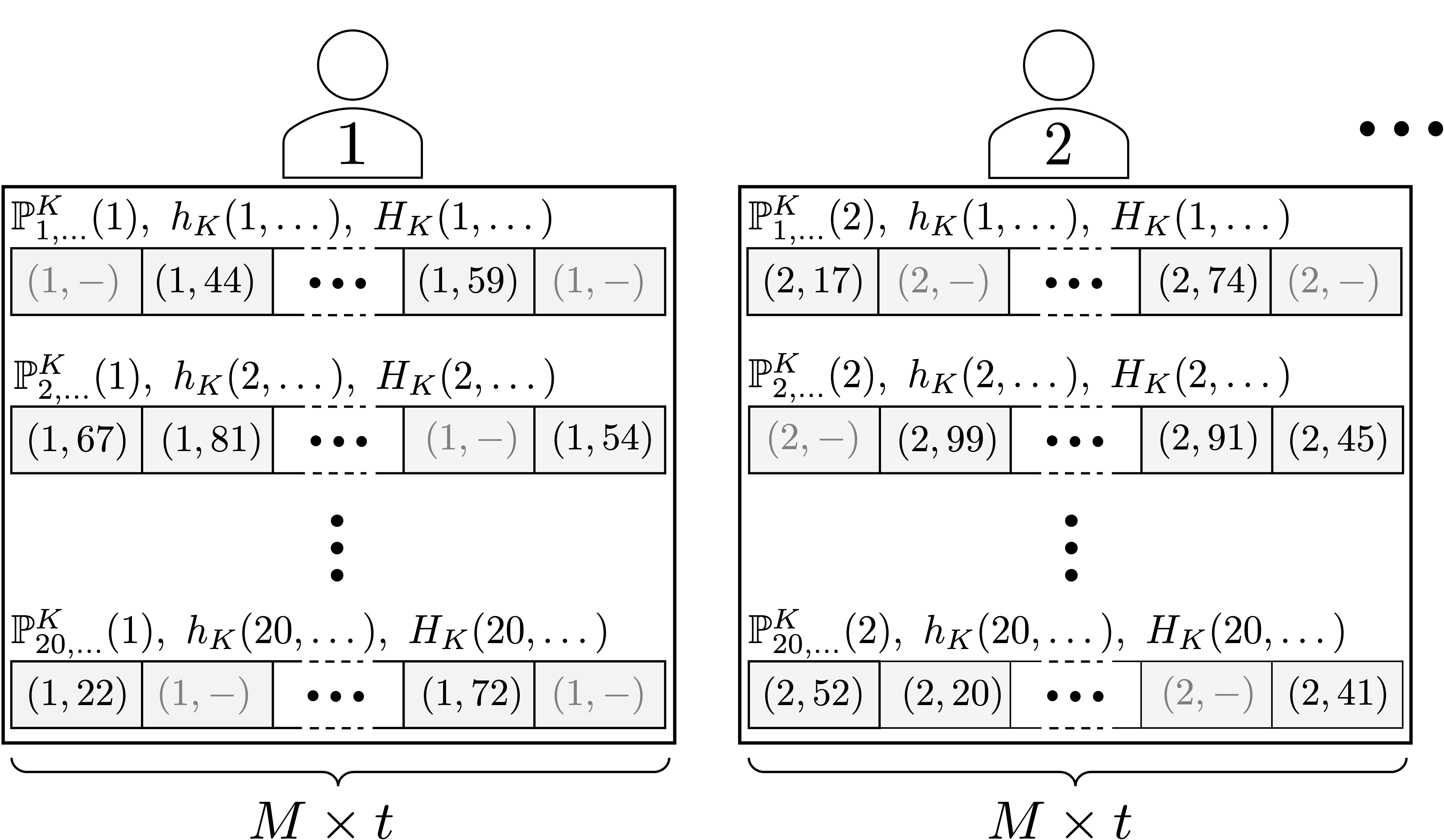}
  \vspace{-0.4cm}
  \caption{The new hashing scheme. Every participant creates multiple tables, $\mathbb{P}$ is the polynomial function, $h_K$ is the keyed hash function used for mapping, and $H_K$ is the keyed hash function used for pseudo-random ordering. The tuples in boxes represent the points on the polynomial, i.e. secret shares, and light gray numbers represent the dummy shares.}
  \label{fig:hashingvisualization}
\end{figure}
\subsection{Protocol Description}
\label{sec:protocoldesc}

Combining the solutions of these challenges, we describe our protocol as follows:

\subsubsection{Non-Interactive Deployment}
Let $K$ be the symmetric key held by all participants $\bm{P_i}\ \forall{i \in \{1, \dots, N\}}$. Let $h_K: \{0,1\}^* \longmapsto [Mt]$ be the keyed hash function used for mapping the elements to the tables. Let $H_K$ be the hash-based message authentication code (HMAC) function with key K, and let $r$ be the id of the current execution of the protocol. 
\begin{enumerate}
    \item Each participant $\bm{P_i}$ creates a table called $Shares$. It has $20$ sub-tables and each sub-table has $Mt$ bins. We denote the $x^{th}$ bin of the $y^{th}$ sub-table as $Shares[y][x]$. Participants fill the tables according to Equation~\ref{eq:share} for each element $s_{i,j}$ in their set and for each $\alpha \in \{1, \dots, 20\}$.
    \begin{alignat}{2}
        \label{eq:share}
         \raisebox{-.5\normalbaselineskip}[0pt][0pt]{\textbf{if}} \quad H_{K}(\alpha, s_{i,j}, r) \leq H_{K}&(\alpha, s, r)\ \  \nonumber \\ 
         \forall s \in \{ s \in S_i \colon h_K(\alpha&,s,r) = h_K(\alpha,s_{i,j},r) \} \ \ \ \\
         \textbf{then} \enspace \ Shares[\alpha][h_K(\alpha,s_{i,j}&,r)] = \mathbb{P}^{K}_{\alpha,s_{i,j},r}(i) \nonumber
    \end{alignat}
    \item All $\bm{P_i}$ fill up the empty bins in $Shares$ with random bits, then send the tables to the Aggregator $\bm{A}$.
    \item After receiving $N$ $Shares$ tables, $\bm{A}$ takes every $t$ combination of $\bm{P_i}$ and checks if the shares in the same bins produce a valid secret.
    \item For each $\bm{P_i}$, $\bm{A}$ sends the indexes of valid reconstructions in the $Shares$ table that $\bm{P_i}$ has sent.
    \item Each participant $\bm{P_i}$ matches the indexes with the underlying elements in their set $S_i$ and get their output $S_i \cap I$.
\end{enumerate}

\subsubsection{Collusion-safe Deployment}
For the collusion-safe deployment, participants do not hold a symmetric key. All the secret shares $\mathbb{P}^{K}_{\alpha,s_{i,j},r}(i)$ are computed with the OPR-SS protocol described in Section \ref{sec:oprss} and instead of the keyed hash functions $H_K$ and $h_K$, the multi key OPRF protocol described in Section \ref{sec:oprf} is used. Since $H_K$ and $h_K$ uses the same inputs for the same elements, a single OPRF call is used to produce both values. All the interactive rounds of the OPR-SS and OPRF protocols for each element is batched together to achieve constant number of interactions. The rest of the protocol is identical with the non-interactive deployment.

\subsection{Privacy of Set Sizes}
Our core protocol do not consider participant's set sizes as private information. So by default, participants communicate their set sizes in plaintext and find the max set size $M$ before running the protocol. If set sizes are private for a particular use case, then $M$ could be decided with a differentially private process. However, this process would need to add positive noise to $M$ since underestimating $M$ will break the core protocol, and this positive noise will adversely affect the performance of the core protocol because the runtime complexity is dependent on $M$.

\section{Analysis of Our Hashing Scheme}
\label{sec:hashing}

In this section, we analyze our new hashing scheme described in Section~\ref{sec:protocoldesc} designed to efficiently address the complexity of the matching problem. The challenge lies in the fact that the secret shares, by definition, do not reveal the underlying element, making it difficult to select a matching combination. In our hashing scheme, each participant places their secret shares into bins of size 1 using a hash function. Consequently, the Aggregator needs to process only one share per participant that fall into the same bin instead of selecting all subsets among a larger set.

In cases of hash collisions, where multiple elements map to the same bin, we must select a single element from those that collide. Rather than choosing this element randomly, we employ an additional hash function to establish an order among the shares and then select the smallest one according to this order. This method increases the likelihood that different participants will choose the same element for a given bin.

Participants repeat this process to create multiple tables using a distinct mapping hash function and ordering hash function for each table. Our goal is to limit the probability of missing an intersection to an acceptably low level by re-randomizing the mapping and ordering for each table. We set our target failure probability to $2^{-40}$ as $40$ is a common statistical security parameter~\cite{de2006cryptographic, boneh2001importance, lindell2016fast}.

We first show that if the size of the table is set to $M\times t$, the probability of missing an intersection is a constant number smaller than 1 for a single table. Let $h: \{0,1\}^* \longmapsto [Mt]$ denote the mapping hash function, $H$ denote the cryptographic ordering hash function, and $shares_i$ denote a single sub-table of the participant $P_i$. 

Given a set element $s_{i,j} \in S_i$ of the participant $P_i$, let $p$ denote the probability of $H(s_{i,k}) < H(s_{i,j})$ for a random $s_{i,k} \in S_i$. Since the output distribution of $H$ is uniform, $p$ can be calculated with the following equation:
\vspace{-0.2cm}
$$p \coloneqq \frac{H(s_{i,j})}{max(H)}$$
The probability that a random set element $s_{i,k}$ maps to the same bin as $s_{i,j}$ is given by
\vspace{-0.3cm}
$$P(h(s_{i,j}) = h(s_{i,k})\ |\ s_{i,j}) = \frac{1}{Mt}$$
The probability that a random set element $s_{i,k}$ blocks $s_{i,j}$, that is, $s_{i,k}$ maps to the same bin and is smaller than $s_{i,j}$ in the ordering:
\vspace{-0.2cm}
\begin{align}
P\Bigl((h(s_{i,j}) & = h(s_{i,k})\ \text{and}\ H(s_{i,k}) < H(s_{i,j})\mid s_{i,j} \Bigr) = \frac{p}{Mt} \nonumber
\end{align}
To calculate a lower bound for the probability that $s_{i,j}$ will be placed in the table, we raise the probability that $s_{i,j}$ is not blocked by a random element to the $M^{th}$ power because $|S_i| \leq M$.

\begin{align}
P&\Bigl(H(s_{i,j}) < H(s)\ \forall s \in \{s \in S_i\ :\ h(s_{i,j}) = h(s)\}\mid s_{i,j}\Bigr) \nonumber \\
&= P(s_{i,j} \in shares_i) = \left(1-\frac{p}{Mt}\right)^{|S_i|-1} \geq \left(1-\frac{p}{Mt}\right)^M \nonumber
\end{align}
With large $M$, the expression $\left(1-\frac{p}{Mt}\right)^M$ approximates $e^{-p/t}$ and for $M>50$ this approximation does not affect any of our results. We now calculate the probability that $t$ different participants will put $s_{i,j}$ in the table given that $s_{i,j}$ is in their sets. This is because, for a successful reconstruction, all participants that have the set element have to put the element into the same table. Since all participants use the same ordering hash function for the same table, the $p$ values will be the same, so we can raise the $e^{-p/t}$ to the $t^{th}$ power to calculate this probability.
\vspace{-0.2cm}
$$P(s_{i,j} \in shares_k\ \forall_{k \in [t]} \mid s_{i,j} \in S_k\  \forall_{k \in [t]}) \geq \left(e^{-p/t}\right)^t = e^{-p}$$
From this equation, it follows that the probability of failure for a particular set element is less than or equal to $1-e^{-p}$, given that this element exists in $t$ different sets.

To calculate the probability of failure for any set element, we treat $p$ as a continuous random variable with a uniform distribution between $0$ and $1$.
\begin{align}
P(\text{fail} \mid p) &\leq 1-e^{-p} \nonumber \\
P(\text{fail}) &\leq \int_{0}^{1}(1-e^{-p})dp = e^{-1} \approx 0.3678 \nonumber
\end{align}
This shows that the probability of failure, more explicitly, the probability of missing any given over-threshold intersection with a single table is at most $e^{-1}$. Every participant needs to create $28$ different tables so that the probability of failure is at most $(e^{-1})^{28} \approx 2^{-40.4}$

We implement two additional optimizations to the hashing scheme that brings the required number of tables to 20. 
First, instead of using a unique ordering hash function for each table, we use the same ordering hash function for every two consecutive tables and reverse the ordering for even-numbered tables.
Second, We do a second insertion after the first insertion for every table. We use a different mapping hash function $h'$ for the second insertion. The second insertion utilize the unused bins from the first insertion. The theoretical analyses of these optimizations are described in Appendix \ref{opt-apdx}.

\section{Evaluation}
\label{sec:eval}
We evaluate the security and performance of our OT-MP-PSI protocol. First, we provide a security analysis using the semi-honest multi-party computation model. We then provide a theoretical complexity analysis for computation and communication costs. Finally, we present experimental performance benchmarks for our protocol, including on the real-world data set from CANARIE IDS Program~\cite{canarie}.

\subsection{Security Analysis}
In this section, we analyze the security of our protocol under the semi-honest model and show that it is secure by constructing the simulators for the participants and the Aggregator.
\begin{theorem}
    \label{thm:security}
    The non-interactive protocol described in Section \ref{sec:protocoldesc} is secure in the semi-honest model given that the Aggregator A is non-colluding.
\end{theorem}

\begin{proof}
    We will prove security by constructing the simulator functions, defined in Section \ref{sec:secmodel}, for the participants and the Aggregator.
    The construction of the simulator for participants \\$SIM_{P_{i}}((S_{i},K,r), I\cap S_i)$ is as follows: The simulator needs to produce the message sent by the Aggregator, which is the indexes of the successful reconstructions as stated in the step 4 of the protocol in Section \ref{sec:protocoldesc}. To do this, simulator creates the $Shares$ table of $P_i$ by replicating the first step of the protocol. Then using the protocol output $I\cap S_i$, it identifies the indexes in the $Shares$ table that contains elements in the intersection. These indexes will be indistinguishable with the $VIEW$ of $P_i$ as they will be identical.
    
    We can construct the simulator $SIM_{A}(r, B)$ for the Aggregator as follows: The simulator creates $N$ different sets $\{S_{1}^{'},\dots,S_{N}^{'}\}$ that satisfy the protocol output $B$ of the Aggregator. For each $j \in [|B|]$, let  tuple $(b_1,\dots,b_N)$ be the $j^{th}$ tuple in $B$, the simulator puts the same random element $x_j$ to all $S_{i}^{'}$ where $b_i = 1$. Then it fills each set with independent uniform random elements until their size is $M$. The simulator then creates the secret shares and the $Shares$ tables for each simulated set, following the protocol description with hash functions parameterized on  a random key $K'$. The simulated $Shares$ tables have the same probability distribution of successful reconstructions as the real ones; secret shares do not reveal any information without a successful reconstruction and the indexes of the tables with successful reconstructions are following the same random distribution. Thus, the simulated $Shares$ tables are indistinguishable from the real ones.
\end{proof}

\begin{theorem}
    \label{thm:security-col}
    The collusion-safe protocol described in Section \ref{sec:protocoldesc} is secure in the semi-honest model given that the at least one key holder is not colluding with the Aggregator.
\end{theorem}

\begin{proof}
    The security proofs of the OPR-SS protocol and the OPRF protocol are shown by Mahdavi et al. \cite{mahdavi2020practical} and Jarecki et al. \cite{jarecki2016} respectively. Thus, we can assume that no participant or the Aggregator learns anything about the additively shared key given at least one key holder is non-colluding. Hence, the simulator's output distributions remain independent given an honest party with secret key information.

    Therefore, we can construct a simulator for a subset of participants colluding with the Aggregator using a simple combination of the simulators described in the previous proof using Goldreich's composition theorem for the semi-honest model~\cite{goldreich}.
\end{proof}

\subsection{Complexity Analysis}
In this section, we provide a theoretical complexity analysis for our OT-MP-PSI protocol. We analyze the computation and communication costs for participants and the Aggregator.

\subsubsection{Computational Complexity}
\label{sec:comp}
\begin{theorem}
    \label{thm:comp}
    The computational complexity of the Aggregator for the protocols described in Section \ref{sec:protocoldesc} is $\mathcal{O}(t^2M \binom{N}{t})$.
\end{theorem}

\begin{proof}
    The Aggregator has to try $\binom{N}{t}$ different combinations of participants. For each combination, the Aggregator has to do $\mathcal{O}(tM)$ Lagrange interpolations. For threshold $t$, a single Lagrange interpolation has a complexity of $\mathcal{O}(t)$. Thus, the total complexity is $\mathcal{O}(t^2M \binom{N}{t})$.
\end{proof}

As a corollary of Theorem \ref{thm:comp}, we can say that the computational complexity for the set intersection problem with two participants $(N=t=2)$ is $\mathcal{O}(M)$, and for the case of a threshold equal to the number of participants $(N=t)$ is $\mathcal{O}(N^2M)$.

\begin{theorem}
    The computational complexity of the participants for the non-interactive protocol described in section \ref{sec:protocoldesc} is $\mathcal{O}(tM)$.
\end{theorem}

\begin{proof}
     The participants have to create $2M$ secret shares for each subtable because of the second insertion optimization described in Appendix \ref{sec:doubleinsert} and $20\cdot 2\cdot M$ secret shares in total. The cost of creating a secret share is $\mathcal{O}(t)$. Thus, the total complexity is $\mathcal{O}(tM)$.
\end{proof}

\subsubsection{Communication Complexity}
\begin{theorem}
    The communication complexity of the non-interactive protocol described in Section \ref{sec:protocoldesc} is $\mathcal{O}(tMN)$.
\end{theorem}

\begin{proof}
    All the participants have to send their $Shares$ tables to the Aggregator. The size of each $Shares$ table is $\mathcal{O}(tM)$. Thus, with $N$ participants, the total communication complexity is $\mathcal{O}(tMN)$.
\end{proof}

\begin{theorem}
    The communication complexity of the collusion-safe protocol described in Section \ref{sec:protocoldesc} is $\mathcal{O}(tkMN)$ and it runs in 5 communication rounds.
\end{theorem}

\begin{proof}
OPR-SS protocol runs in three communication rounds and it has $\mathcal{O}(tk)$ communication complexity where $k$ is the number of key holders as shown by Mahdavi et al. \cite{mahdavi2020practical}. OPRF protocol runs in one communication round and it has constant communication cost. There will be $20\cdot 2\cdot MN$ OPR-SS and OPRF invocations in the entire protocol. Therefore, the total communication cost of the protocol will be $\mathcal{O}(tkMN)$ All of the OPR-SS invocations can be batched together and after getting the secret shares, all of the OPRF invocations can also be batched. Thus, adding the one additional communication round between the participants and the aggregator, the protocol will run in 5 communication rounds in total.
\end{proof}

\subsection{Correctness Evaluation}

\begin{figure}[h!]
    \centering
    \begin{tikzpicture}
        \begin{semilogyaxis}[
            xlabel={Number of Tables},
            ylabel style={align=center}, ylabel={Number of Missed \\ Intersection Elements},
            legend pos=north east,
            grid=major,
            width=0.42\textwidth,
            height=0.28\textwidth,
        ]
        \addplot[ybar, fill=YellowGreen, YellowGreen, ybar legend] table [x index=0, y index=1, col sep=comma] {simulate-share-array-4.csv}; \addlegendentry{Experimental Results}

        \addplot[ForestGreen, mark=square*] table [x index=0, y index=2, col sep=comma] {simulate-share-array-4.csv};
        
        \legend{Experimental Results, Computed Upper Bound}
        \end{semilogyaxis}
    \end{tikzpicture}
    \vspace{-0.2cm}
    \caption{Number of missed Intersection Elements in $10^7$ trials ($M=200, t=4$)}
    \label{fig:fail-eval}
\end{figure}
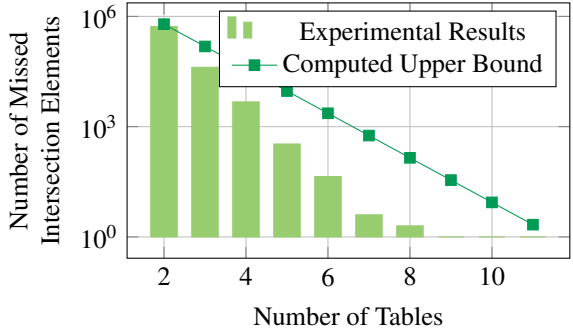

We experimentally evaluate the failure rate of our hashing scheme to validate our theoretical failure probability analysis. Figure \ref{fig:fail-eval} shows the number of missed over-threshold intersection elements out of 10 million trials for different number of tables. For odd number of tables the last table does not have a pair, so the failure probability upper bound is calculated as $(\text{2 table fail probability})^{(i-1)/2} \times (\text{1 table fail probability})$ where $i$ is the number of tables. It can be observed that the experimental results are well below the computed upper bounds. 

It is not computationally feasible to experimentally show the $2^{-40}$ lower bound that we claim in Section \ref{sec:hashing}. However, these results are still helpful to show the soundness of our calculations.

\subsection{Performance Evaluation}
In this section, we provide performance benchmarks for our OT-MP-PSI protocol. The protocol has two distinct phases that occur in sequential order: share creation in the participants and reconstruction in the Aggregator. Thus, we measure their runtime separately using synthetic data. We compare our results with the implementation of Mahdavi et al. \cite{mahdavi2020practical} as the only other solution with a public code repository at the time of writing. We also evaluate our system using real-world network data obtained from CANARIE IDS Program~\cite{canarie}.

\subsubsection{Setup}
We implement our protocol using Julia language with 430 lines of code. The cryptographic libraries that we use are SHA.jl and Nettle.jl. We use Julia's threads library for parallelization. We used the 61-bit Mersenne prime for the finite field. This is so that we can take advantage of fast modulo arithmetic with Mersenne primes and use 128-bit integers instead of arbitrary precision integers. Our implementation and the benchmark scripts used to generate the graphs in Section \ref{sec:eval} are available at \url{https://github.com/onurerenarpaci/Rand-Hashing-OT-MP-PSI}. The data collected from the CANARIE IDS Program is not disclosed due to the privacy agreements between institutions.

We use a server with 8 x Intel Xeon E7-8870 processors for all of our experiments. In total, we have 80 physical cores, each running at 2.4 GHz. The server has 1 TB of memory. However, none of our experiments require more than 14 GB of memory. Most of the memory is consumed by parallel Lagrange interpolations, and the degree of parallelism depends on the number of available CPU cores.

\subsubsection{Reconstruction}
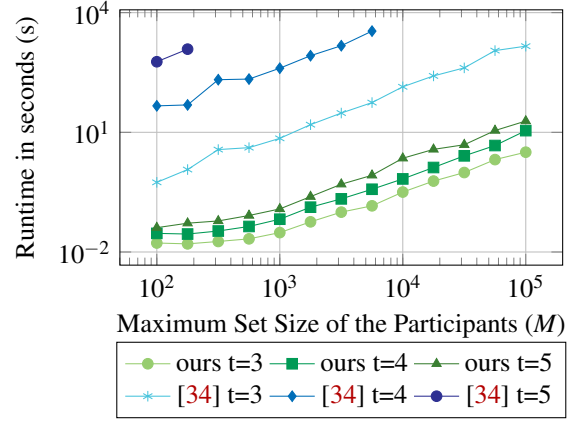
\begin{figure}[h]
    \centering
    \begin{tikzpicture}
        \begin{loglogaxis}[
            xlabel={Maximum Set Size of the Participants ($M$)},
            ylabel={Runtime in seconds (s)},
            legend style={at={(0.5,-0.3)}, anchor=north,legend columns=3},
            grid=major,
            width=0.42\textwidth,
            height=0.28\textwidth,
        ]
        \addplot[YellowGreen, mark=*] table [x index=0, y index=1, col sep=comma] {recon-comp.csv};
        \addplot[ForestGreen, mark=square*] table [x index=0, y index=2, col sep=comma] {recon-comp.csv};
        \addplot[OliveGreen, mark=triangle*] table [x index=0, y index=3, col sep=comma] {recon-comp.csv};
        \addplot[SkyBlue, mark=asterisk] table [x index=0, y index=4, col sep=comma] {recon-comp.csv};
        \addplot[NavyBlue, mark=diamond*] table [x index=0, y index=5, col sep=comma] {recon-comp.csv};
        \addplot[Blue, mark=*] table [x index=0, y index=6, col sep=comma] {recon-comp.csv};

        \legend{ours t=3, ours t=4, ours t=5, \cite{mahdavi2020practical} t=3, \cite{mahdavi2020practical} t=4, \cite{mahdavi2020practical} t=5}
        \end{loglogaxis}
    \end{tikzpicture}
    \vspace{-0.2cm}
    \caption{Reconstruction time comparison with Mahdavi et al. \cite{mahdavi2020practical} $(N=10)$}
    \label{fig:mahd_comp}
\end{figure}

We compare the reconstruction time of our protocol with the reconstruction time of Mahdavi et al. \cite{mahdavi2020practical} in Figure \ref{fig:mahd_comp}. We terminated some of the experiments using Mahdavi et al.'s work~\cite{mahdavi2020practical} because they ran for more than an hour. Our protocol is at least two orders of magnitude faster than Mahdavi et al. \cite{mahdavi2020practical} for all parameters, and the difference increases exponentially with bigger threshold values.

\begin{figure}[h]
    \centering
    \begin{tikzpicture}
        \begin{axis}[
            date coordinates in=x,
            xlabel={Date},
            ylabel={Runtime in seconds (s)},
            legend pos=south east,
            grid=major,
            width=0.42\textwidth,
            height=0.28\textwidth,
            xticklabel=\day.\month,
        ]
        \addplot[YellowGreen, mark=*, mark size=1pt] table [x index=5, y index=3, col sep=comma] {canarie-tick6-converted.csv};
        \end{axis}
    \end{tikzpicture}
    \vspace{-0.2cm}
    \caption{Reconstruction time on CANARIE IDS data~\cite{canarie}}
    \label{fig:canarie}
\end{figure}
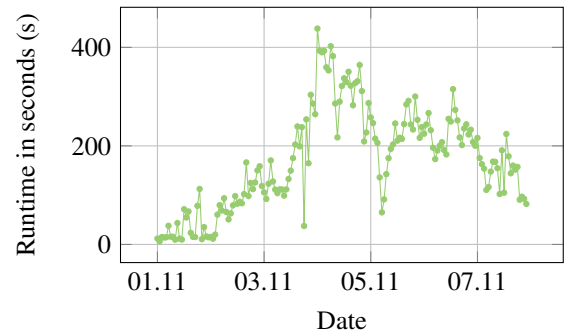

To demonstrate the practicality of our protocol in a real-world setting, we collected network connection logs from 54 institutions via the CANARIE IDS Program for a one-week period (November 1–8, 2023). During this time, the CANARIE IDS Program ingested approximately 8 billion logs per day, totaling about 700 GB (gzip-compressed) daily. We filtered the logs to only include records where the source was an external IP address and the destination was an internal IP address. Following Zabarah et al.~\cite{zabarah2023approach}, we ran the OT-MP-PSI protocol on hourly batches. For each institution and each hour, we extracted the unique external IP addresses connecting to that institution during that hour to form the participant sets for the protocol. If an institution has no external IPs that start connections in that hour, the institution is not included in the protocol. We set the threshold value to 3 which is the suggested value by Zabarah et al.~\cite{zabarah2023approach}. The reconstruction time during this week is shown in Figure~\ref{fig:canarie}. The maximum reconstruction time during the week is 438 seconds, where 40 institutions participate in the protocol, and the maximum set size is 220,011. The mean and median reconstruction times are 170 and 168 seconds, the mean and median participating institution counts are 33 and 32, and the mean and median maximum set sizes are 144,045 and 162,113 respectively. Since the protocol runs once every hour, we consider the mean/median reconstruction time acceptable for a real-world application.

\begin{figure}[h]
    \centering
    \begin{tikzpicture}
        \begin{semilogyaxis}[
            xlabel={Number of Participants ($N$)},
            ylabel={Runtime in seconds (s)},
            legend pos=south east,
            legend style={fill=white, fill opacity=0.6, draw opacity=1,text opacity=1},
            grid=major,
            width=0.42\textwidth,
            height=0.28\textwidth,
        ]
        \addplot[YellowGreen, mark=*] table [x index=0, y index=1, col sep=comma] {recon-N.csv};
        \addplot[ForestGreen, mark=square*] table [x index=0, y index=2, col sep=comma] {recon-N.csv};
        \addplot[OliveGreen, mark=triangle*] table [x index=0, y index=3, col sep=comma] {recon-N.csv};

        \legend{t=3, t=4, t=5, }
        \end{semilogyaxis}
    \end{tikzpicture}
    \vspace{-0.2cm}
    \caption{Reconstruction time versus number of participants $(M=10000)$}
    \label{fig:num_participants}
\end{figure}
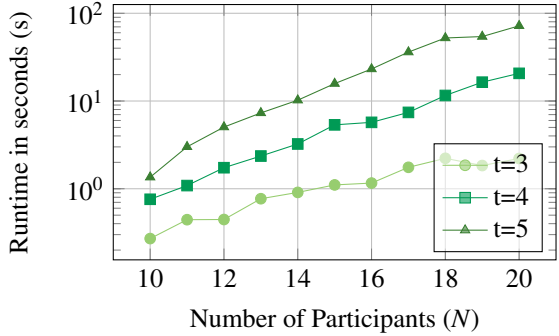

We measure the reconstruction time of our protocol for different numbers of participants in Figure \ref{fig:num_participants}. The reconstruction time increases polynomially with the number of participants. This is because of the $\binom{N}{t}$ term in our protocol's complexity has an upper bound of $(\frac{eN}{t})^t > \binom{N}{t}$.

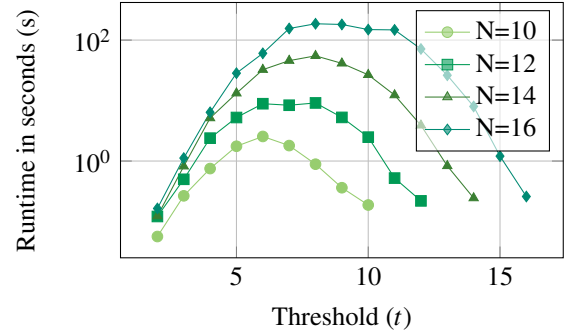
\begin{figure}[h]
    \centering
    \begin{tikzpicture}
        \begin{semilogyaxis}[
            xlabel={Threshold ($t$)},
            ylabel={Runtime in seconds (s)},
            legend pos=north east,
            legend style={fill=white, fill opacity=0.6, draw opacity=1,text opacity=1},
            grid=major,
            width=0.42\textwidth,
            height=0.28\textwidth,
        ]
        \addplot[YellowGreen, mark=*] table [x index=0, y index=1, col sep=comma] {recon-t.csv};
        \addplot[ForestGreen, mark=square*] table [x index=0, y index=2, col sep=comma] {recon-t.csv};
        \addplot[OliveGreen, mark=triangle*] table [x index=0, y index=3, col sep=comma] {recon-t.csv};
        \addplot[PineGreen, mark=diamond*] table [x index=0, y index=4, col sep=comma] {recon-t.csv};

        \legend{N=10, N=12, N=14, N=16}
        \end{semilogyaxis}
    \end{tikzpicture}
    \vspace{-0.2cm}
    \caption{Reconstruction time vs threshold $(M=10000)$}
    \label{fig:threshold}
\end{figure}

We measure the reconstruction time of our protocol for different threshold values in Figure \ref{fig:threshold}. The reconstruction time increases exponentially with the threshold value until $t=N/2$, after which it decreases exponentially. This is because of the $\binom{N}{t}$ term in the complexity of our protocol.

\subsubsection{Share Generation}

\begin{figure}[h]
    \centering
    \begin{tikzpicture}
        \begin{loglogaxis}[
            xlabel={Maximum Set Size of the Participants ($M$)},
            ylabel={Runtime in seconds (s)},
            legend pos=north west,
            legend style={fill=white, fill opacity=0.6, draw opacity=1,text opacity=1},
            grid=major,
            width=0.42\textwidth,
            height=0.28\textwidth,
        ]
        \addplot[SkyBlue, mark=*] table [x index=0, y index=1, col sep=comma] {sg-comp.csv};
        \addplot[NavyBlue, mark=square*] table [x index=0, y index=2, col sep=comma] {sg-vs-rc-2.csv};
        \addplot[YellowGreen, mark=triangle*] table [x index=0, y index=3, col sep=comma] {sg-comp.csv};
        \addplot[OliveGreen, mark=diamond*] table [x index=0, y index=4, col sep=comma] {sg-comp.csv};

        \legend{col-safe t=3, col-safe t=6, non-int t=3, non-int t=6}
        \end{loglogaxis}
    \end{tikzpicture}
    \vspace{-0.2cm}
    \caption{Share generation time of a single participant with the collusion-safe and non-interactive deployment}
    \label{fig:share-gen}
\end{figure}
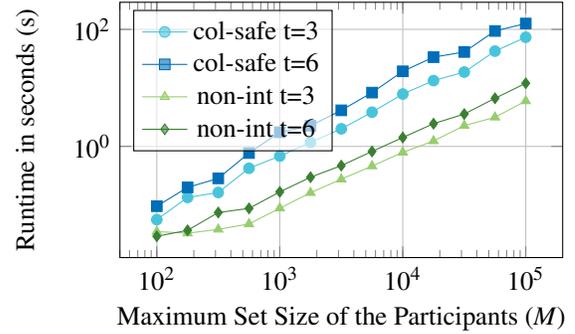

We measure the share generation time of a single participant for different maximum set sizes and different threshold $t$ values with the collusion-safe deployment and the non-interactive deployment in Figure \ref{fig:share-gen}. The share generation time increases linearly with the maximum set size, which confirms the $\mathcal{O}(tM)$ complexity of the participants. We observe that the collusion-safe deployment is approximately an order of magnitude slower than the non-interactive version.

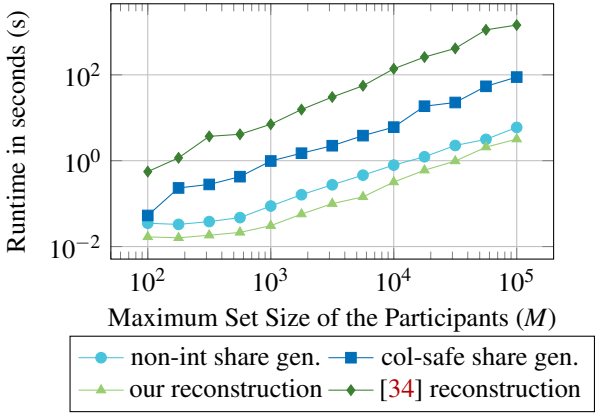
\begin{figure}[h]
    \centering
    \begin{tikzpicture}
        \begin{loglogaxis}[
            xlabel={Maximum Set Size of the Participants ($M$)},
            ylabel={Runtime in seconds (s)},
            legend style={at={(0.5,-0.3)}, anchor=north,legend columns=2},
            grid=major,
            width=0.42\textwidth,
            height=0.28\textwidth,
        ]
        \addplot[SkyBlue, mark=*] table [x index=0, y index=3, col sep=comma] {sg-comp.csv};
        \addplot[NavyBlue, mark=square*] table [x index=0, y index=2, col sep=comma] {sg-comp.csv};
        \addplot[YellowGreen, mark=triangle*] table [x index=0, y index=3, col sep=comma] {sg-vs-rc-2.csv};
        \addplot[OliveGreen, mark=diamond*] table [x index=0, y index=4, col sep=comma] {sg-vs-rc-2.csv};

        \legend{non-int share gen., col-safe share gen., our reconstruction, \cite{mahdavi2020practical} reconstruction}
        \end{loglogaxis}
    \end{tikzpicture}
    \vspace{-0.2cm}
    \caption{Comparison of reconstruction of the Aggregator and share generation of a single participant (t=3)}
    \label{fig:sg-vs-rc}
\end{figure}

We also compare the runtimes of share generation process and the reconstruction process in Figure \ref{fig:sg-vs-rc}. We see that our new-hashing algorithm shifted the bottleneck from reconstruction to share generation.

\section{Related Work}
\label{sec:related}

We discuss the previous solutions to the OT-MP-PSI problem, some other closely related problems, and related work regarding our new hashing scheme.

\subsection{Previous Solutions to OT-MP-PSI}

\begin{table*}[h]
\centering
\caption{Comparison of OT-MP-PSI Solutions}
\begin{tabular}{|>{\centering\arraybackslash}m{3.5cm}|>{\centering\arraybackslash}m{2.8cm}|>{\centering\arraybackslash}m{2.9cm}|>{\centering\arraybackslash}m{2.4cm}|>{\centering\arraybackslash}m{3.7cm}|}
\hline
\vspace{0.1cm} \textbf{Solution} & \textbf{Comp. Complexity} & \textbf{Comm. Complexity} & \textbf{Comm. Rounds} & \textbf{Collusion Resistance} \\ \hline
\vspace{0.15cm} Kissner and Song~\cite{kissner2004private} \vspace{0.05cm}  & $\mathcal{O}(N^3M^3)$ & $\mathcal{O}(N^3M)$ & $\mathcal{O}(N)$ & up to $k$ collusions \\ \hline
\vspace{0.15cm} Mahdavi et al.~\cite{mahdavi2020practical} \vspace{0.05cm}  & $\mathcal{O}(M(N \log{M}/t)^{2t})$ & $\mathcal{O}(tMNk)$ & $\mathcal{O}(1)$ & up to $k$ collusions\\ \hline
\vspace{0.15cm} Ma et al.~\cite{ma2024over} \vspace{0.05cm} & $O(N|S|)$ & $O(N|S|)$ & $\mathcal{O}(1)$ & two non-colluding server \\ \hline
\vspace{0.15cm} Ours (Non-interactive) \vspace{0.06cm} & $\mathcal{O}(t^2M \binom{N}{t})$ & $\mathcal{O}(tMN)$ & $1$ & non-colluding server \\ \hline
\vspace{0.15cm} Ours (Collusion-safe) \vspace{0.06cm} & $\mathcal{O}(t^2M \binom{N}{t})$ & $\mathcal{O}(tMNk)$ & $\mathcal{O}(1)$ & up to $k$ collusions \\ \hline
\end{tabular}
\label{tab:comparison}
\end{table*}

\subsubsection{Kissner and Song}

Kissner and Song proposed the first solution to the OT-MP-PSI problem~\cite{kissner2004private}. They use polynomial rings to represent multisets. Given a multiset $S = \{S_j\}_{1 \leq j \leq k}$, they represent it as a polynomial $f(x) = \Pi_{j=1}^{k} (x - S_j)$. With this representation, the union of two sets becomes the product of two polynomials. They use homomorphic encryption to encrypt the polynomials; then, all the participants sequentially multiply their sets to produce a polynomial that represents the union of all sets. Then, they do homomorphic derivative operations on this polynomial to find the elements that exist in at least threshold number of sets. Their solution requires a high number of communication rounds. Specifically, their protocol requires $\mathcal{O}(N)$ communication rounds and $\mathcal{O}(N^3M)$ communication complexity. Moreover, multiplying the polynomials and evaluating them for each set element has a computational complexity of $\mathcal{O}(N^2M^3)$ for each participant, totaling to $\mathcal{O}(N^3M^3)$ computation when serialized. There are three reasons why this is worse than our $\mathcal{O}(t^2M \binom{N}{t})$ complexity. First, every arithmetic operation is performed with homomorphic encryption, which is computationally expensive. Second, each polynomial multiplication depends on the previous participants' multiplication, making it less parallelizable than our reconstruction process. Finally, in the use cases we are interested, M is much larger than N, i.e, few users with large datasets, resulting in $t^2\binom{N}{t} \ll N^3M^2$. However, if the threshold $t$ is close to $N/2$ and $N$ is large enough, this assumption might not hold.

\subsubsection{Mahdavi et al.}
Mahdavi et al.~\cite{mahdavi2020practical} introduce a solution with a constant number of rounds and better communication complexity than Kissner and Song~\cite{kissner2004private}. They use Shamir's secret sharing scheme to reveal identical items over the threshold, similar to our solution. They introduce the OPR-SS protocol that we use for our collusion-safe deployment option.

Their scheme has $\mathcal{O}(M(N \log{M}/t)^{2t})$ computational complexity due to the binning technique, which we substantially improve with our new hashing scheme.

\subsubsection{Ma et al.}
Ma et al.~\cite{ma2024over} propose a solution designed for use cases involving small input sets and small domain sizes. Their approach employs two mutually non-colluding servers to reduce the OT-MP-PSI problem to a two-party computation (2PC) problem. Uniquely, in this protocol, servers can compute results for multiple thresholds at no extra client cost. However, the main limitation of this solution is its computational and communication complexity, which is $\mathcal{O}(N|S|)$, where $|S|$ is the size of the domain of the elements. Therefore, if the inputs come from a large domain, such as the IPv6 addresses in our use case, this solution would not be feasible.

\subsubsection{Threshold Private Set Intersection (TPSI)}
Although very similarly named, the OT-MP-PSI and TPSI problems are very different. We can observe this similar naming caused some confusion in the existing literature: Mohanty et al. \cite{mohanty2024quantum} incorrectly classify \cite{kissner2004private, mahdavi2020practical, bay2022practical} as TPSI protocols and compare them with other TPSI protocols.

TPSI \cite{zhao2018can} is a cryptographic problem where two or more participants holding sets of elements want to learn the intersection of their sets if the intersection size is above a certain threshold. Several solutions have been proposed for the TPSI problem \cite{zhao2018can, hallgren2017privatepool, ghosh2019communication, badrinarayanan2021multi}.

\subsubsection{Quorum-PSI}
Quorum-PSI, introduced by Chandran et al. \cite{nishanth2021}, is a less general version of the OT-MP-PSI problem. In Quorum-PSI there is a dedicated party $P_1$ (the only party with output) which has the element in the intersection also in its set, i.e., not all possible subsets of parties that exceed the threshold are considered, but only those that contain $P_1$. This difference enables more efficient solutions \cite{nishanth2021}, which are not possible for the OT-MP-PSI problem.

\subsection{Related Work Regarding the Hashing Scheme}

Although there are no previous works that directly address the OT-MP-PSI problem with a hashing scheme, there are some works that address similar problems in the context of two-party private set intersection (2P-PSI) and multi-party private set intersection (MP-PSI) protocols, which are special cases of OT-MP-PSI where $N=t=2$ for 2P-PSI and $t=N$ for MP-PSI. Both of these problems have many real-world use cases, such as path discovery in social networks~\cite{ghita2009privacy}, botnet detection~\cite{Nagaraja2010Bot}, web ad campaign performance measurement~\cite{pinkas2014faster}, or cheater detection in games~\cite{Bursztein2011Open}. Our hashing scheme efficiently extends to both problems because our computational complexity $\mathcal{O}(t^2M \binom{N}{t})$ becomes $\mathcal{O}(M)$ for $N=t=2$ and $\mathcal{O}(N^2M)$ for $t=N$ as shown in Section \ref{sec:comp}.  

\subsubsection{Cuckoo Hashing with Simple Hashing}
\textbf{Cuckoo Hashing.} Cuckoo hashing~\cite{cuckoo} is a technique to avoid hash collisions by assigning each key multiple possible locations using different hash functions. Originally designed for worst-case constant-time lookups, it has since been applied in private protocols such as private set intersection (PSI)~\cite{pinkas2018efficient, pinkas2020paxos}, private information retrieval (PIR)~\cite{ali2021Comm, angel2018pir}, oblivious RAM (ORAM)~\cite{asharov2020opt, goodrich2011privacy}, and symmetric searchable encryption (SSE)~\cite{bossuat2021sse, patel2019mitig}. While many other variants have been proposed~\cite{arbitman2009deam, kirsch2010more, fotakis2005space}, the original construction by Pagh and Rodler~\cite{cuckoo} uses two hash functions $h_0, h_1$ mapping $M$ elements to two tables $T_0, T_1$, each with $(1 + \varepsilon)M$ bins, storing at most one element per bin. When inserting $x$ into bin $h_b(x)$ in table $T_b$, if another element $y$ already occupies that bin, $y$ is evicted to $h_{1-b}(y)$ in $T_{1-b}$, with $b$ toggled each time. The process continues until no collisions remain or a relocation limit is reached.

A naive PSI approach would have both parties, Alice and Bob, use Cuckoo hashing and compare items in corresponding bins. However, an element $x$ might be placed by Alice in $h_0(x)$ but by Bob in $h_1(x)$, and comparing both bins would reveal to Bob that Alice holds an item corresponding to those two bins, breaking privacy. The solution, as used in~\cite{freedman2016efficient, pinkas2014faster, pinkas2015phasing}, has Alice use Cuckoo hashing while Bob uses simple hashing, mapping each of his items to both bins $h_0(x)$ and $h_1(x)$. Bob pads each bin with $b = O(\log M / \log \log M)$ elements, resulting in $O(M \log M / \log \log M)$ total comparisons. This construction applies only to two-party protocols.

\subsubsection{2D Cuckoo Hashing}
Pinkas et al. \cite{pinkas2018efficient} propose a new variant of Cuckoo hashing called 2D Cuckoo Hashing that reduces the number of comparisons to $O(M)$ for 2P-PSI. The construction involves two tables with two sub-tables in each table. Alice and Bob use different strategies to assign their items to the tables. These strategies are modified versions of the Cuckoo hashing insertion strategy designed in a way to ensure that if Alice and Bob have the same element, there will be a sub-table that has the element in both participants.

Although both our hashing scheme and the 2D Cuckoo hashing achieve $O(M)$ complexity for 2P-PSI, our scheme is more general and can be applied to any number of participants and any threshold $t$. 2D Cuckoo hashing is specifically designed for 2P-PSI and does not generalize to more than two participants or a threshold $t$ due to Alice and Bob's asymmetric insertion strategy. However, being this specific allows 2D Cuckoo hashing to achieve lower constants in concrete implementations for two parties.

\section{Conclusion}
\label{sec:conclusion}

In this work, we introduced a practical Over-Threshold Multiparty Private Set Intersection (OT-MP-PSI) protocol designed specifically to support collaborative network intrusion detection. Our protocol efficiently identifies IP addresses appearing across multiple institutions' logs, allowing timely detection of coordinated attacks without compromising privacy.

Previous approaches, such as those by Kissner and Song~\cite{kissner2004private} and Mahdavi et al.~\cite{mahdavi2020practical}, either required extensive communication rounds or incurred prohibitive computational costs for large datasets. Our solution addresses these limitations, achieving a computational complexity of $\mathcal{O}(t^2M\binom{N}{t})$ through a novel hashing scheme, thus making it practical for real-world network security scenarios.

Our experimental evaluation, using real-world network logs from CANARIE IDS Program\cite{canarie}, demonstrates the practicality and efficiency of our protocol. With an average runtime of just 170 seconds for analyzing data from 33 institutions and datasets containing over 144,000 IP addresses, our solution is well-suited for operational deployment in network security contexts.

In summary, our protocol provides a scalable, privacy-preserving mechanism essential for collaborative intrusion detection. Future work will explore optimizations for efficiently handling participant combinations, further enhancing performance in large-scale deployments.

\section*{Acknowledgments}
We gratefully acknowledge the support of NSERC for grants RGPIN-2023-03244, RGPIN-06587-2019, IRC-537591, the Government of Ontario, and the Royal Bank of Canada for funding this research, and CANARIE for providing the network logs used in the experiments.

\bibliographystyle{plain}
\bibliography{references.bib}

\appendix

\section{Hashing Scheme Additional Optimizations}
\label{opt-apdx}
\subsection{Reversing the ordering}
Instead of using a unique ordering hash function for each table, using the same ordering hash function for every two consecutive tables and reversing the ordering for even-numbered tables gives better results. The intuition behind this approach is that we know the ``unlucky'' elements in the odd-numbered tables; these are elements with big $p$ values. By reversing the order and making them ``lucky'' in the next consecutive table, we get better probabilities than re-randomizing for each table.

With this technique, the value of $p$ for a particular element will be equal to $1-p$ in the next table. Therefore, the probability of missing a particular intersection in two consecutive tables is given by
$$P(\text{fail with two tables} \mid p) \leq (1-e^{-p})(1-e^{-(1-p)})$$
We then calculate the probability of failure for any intersection similarly to the single table calculation.
\begin{align}
P(\text{fail with two tables}) &\leq \int_{0}^{1}(1-e^{-p})(1-e^{-(1-p)})dp \nonumber\\
&= 3e^{-1} - 1 \approx 0.10363 \nonumber
\end{align}
This technique reduces the required number of tables to $26$ to obtain $(3e^{-1} - 1)^{13} \approx 2^{-42.5}$ failure probability.

\subsection{Utilizing the empty bins}
\label{sec:doubleinsert}

Normally, after inserting their elements into the tables, the participants would fill the empty bins with dummy shares. However, we can utilize this wasted space to reduce the number of required tables. With this optimization, participants will do a second insertion after the first insertion for every table. They will use a different mapping hash function $h'$ for the second insertion and reverse the ordering hash function from the first insertion. While doing the second insertion, elements from the first insertion have priority in the table; in other words, if an element maps to an already occupied position, we do not make any changes to the table. 

Therefore, two events must occur for an element to be successfully put into the table in the second insertion: it must map into an empty bin and be the smallest element among the elements that map into the same bin. We already know the probability of the second event occurring in $t$ different sets; it is $e^{-(1-p)}$ due to reverse ordering. To calculate the probability of the first event, we first calculate the probability of any bin being empty after the first insertion. We do this by calculating the probability of a single element not mapping to that bin repeated $M$ times. Again, the approximation does not affect the results for $M>50$.
\[
P(shares_i[x]\ \text{is empty}) = \left(\frac{tM - 1}{tM}\right)^{|S_i|} \geq \left(\frac{tM - 1}{tM}\right)^M \approx e^{-1/t}
\]
We then calculate the probability of $s_{i,j}$ being mapped into empty bins in all $t$ sets in the second insertion, given that $s_{i,j}$ exists in $t$ different sets.
\begin{align}
P\Bigl(shares_k[h'(s_{i,j})]\ &\text{is empty}\ \forall_{k \in [t]} \mid s_{i,j} \in S_k\ \forall_{k \in [t]}\Bigr) \nonumber\\ 
&\geq \left(e^{-1/t}\right)^t = e^{-1} \nonumber
\end{align}
\noindent
Now, we multiply these two probabilities to find the second insertion success probability of a particular intersection.
$$P(\text{success in second insertion} \mid p) \geq e^{-(1-p)}\times e^{-1} = e^{p-2}$$
\noindent
Lastly, the probability of missing a particular intersection in a single table is given by the failure probability of the first insertion multiplied by the failure probability of the second insertion. Afterward, we calculate the failure probability for any intersection.
\begin{align}
P(\text{fail} \mid p) \leq& (1-e^{-p})(1-e^{p-2}) \nonumber\\
P(\text{fail}) \leq& \int_{0}^{1}(1-e^{-p})(1-e^{p-2})dp = 2e^{-2} \approx 0.2706 \nonumber
\end{align}
\noindent
This approach reduces the required number of tables to $22$ to have $\left(2e^{-2}\right)^{22} = 2^{-41.5}$ failure probability. We also explored the strategy of doing multiple insertions until the table is filled. However, for use cases where $t$ is significantly smaller than $M$ this approach results in insignificant improvements.

Finally, we can get better results by combining the two optimizations; more explicitly, we will do second insertions for every table and order reversing for every two consecutive tables. In this case, the failure probability for a particular intersection in two consecutive tables is given by

\begin{align}
&P(\text{fail with 2 tables} \mid p) \nonumber\\ 
&\leq \eqnmarkbox[gray]{farr}{(1-e^{-p})(1-e^{p-2})}\ \eqnmarkbox[gray]{sarr}{(1-e^{-(1-p)})(1-e^{-p-1})} \nonumber
\end{align}

\annotate[yshift=-1em, color=black]{below}{farr}{first table}
\annotate[yshift=-1em, color=black]{below}{sarr}{second table}
\vspace{0.5cm}

\noindent
We then calculate the failure probability for any intersection

\begin{align}
&P(\text{fail with 2 tables}) \nonumber \\
&\leq \int_{0}^{1}(1-e^{-p})(1-e^{p-2})(1-e^{-(1-p)})(1-e^{-p-1})dp \nonumber \\
&= 2e^{-1} + 2e^{-2} + 3e^{-4} - 1\approx 0.06138 \nonumber
\end{align}

The combination of these two approaches reduces the number of required tables to $20$ to obtain a failure probability of $(0.06138)^{10}= 2^{-40.3}$.
\end{document}